\documentclass[11pt,reqno]{amsart}
\usepackage{amssymb}
\usepackage{amsfonts}
\usepackage{amsmath}
\usepackage{stmaryrd}
\usepackage{physics}
\usepackage{braket}
\usepackage{dsfont}
\usepackage{bbold}
\usepackage{graphicx}
\usepackage{relsize}
\usepackage{enumerate}
\usepackage[margin=1in]{geometry}
\usepackage{enumitem}
\usepackage{mathtools}
\usepackage{bbm}
\usepackage{graphbox}
\usepackage{comment}
\usepackage{float}
\usepackage[font=scriptsize]{caption}
\usepackage{booktabs}
\usepackage[table,xcdraw]{xcolor}
\usepackage{hyperref}
\hypersetup{colorlinks = true, linkcolor = blue, urlcolor  = blue, citecolor = red}
\renewcommand{\epsilon}{\varepsilon}
\renewcommand{\phi}{\varphi}
\newcommand{\overbar}[1]{\mkern 1.5mu\overline{\mkern-1.5mu#1\mkern-1.5mu}\mkern 1.5mu}
\newcommand{\iden}{\mathbb{1}}
\newcommand{\CLDUI}{\mathsf{CLDUI}}
\newcommand{\LDUI}{\mathsf{LDUI}}
\newcommand{\LDOI}{\mathsf{LDOI}}

\newcommand{\MLDUI}[1]{\mathcal{M}_{#1}(\mathbb{C})^{\times 2}_{\mathbb{C}^{#1}}}
\newcommand{\MLDOI}[1]{\mathcal{M}_{#1}(\mathbb{C})^{\times 3}_{\mathbb{C}^{#1}}}
\newcommand{\M}[1]{\mathcal{M}_{#1}(\mathbb{C})}
\newcommand{\C}[1]{\mathbb{C}^{#1}}
\newcommand{\Mreal}[1]{\mathcal{M}_{#1}(\mathbb{R})}

\newcommand{\U}[1]{\mathcal{U}(#1)}
\newcommand{\UU}[1]{\mathcal{U}(#1\otimes #1)}
\renewcommand{\O}[1]{\mathcal{O}(#1)}
\newcommand{\OO}[1]{\mathcal{O}(#1\otimes #1)}

\newtheorem{theorem}{Theorem}[section]
\newtheorem{definition}[theorem]{Definition}
\newtheorem*{definition*}{Definition}
\newtheorem{proposition}[theorem]{Proposition}
\newtheorem{remark}[theorem]{Remark}
\newtheorem{corollary}[theorem]{Corollary}
\newtheorem{lemma}[theorem]{Lemma}

\newtheorem*{conjecture*}{Conjecture}

\theoremstyle{definition}

\definecolor{darkgreen}{rgb}{0,0.392,0}

\author{Satvik Singh}
\email{satviksingh2@gmail.com}
\address{\parbox{\linewidth}{Department of Applied Mathematics and Theoretical Physics, \\ University of Cambridge, Cambridge, United Kingdom.}}

\author{Ion Nechita}
\email{ion.nechita@univ-tlse3.fr}
\address{Laboratoire de Physique Th\'eorique, Universit\'e de Toulouse, CNRS, UPS, France.}

\title[Diagonal unitary and orthogonal symmetries in quantum theory II]{Diagonal unitary and orthogonal symmetries \\ in quantum theory II: evolution operators}

\begin{document}

\begin{abstract}
    We study bipartite unitary operators which stay invariant under the local actions of diagonal unitary and orthogonal groups. We investigate structural properties of these operators, arguing that the diagonal symmetry makes them suitable for analytical study. As a first application, we construct large new families of dual unitary gates in arbitrary finite dimensions, which are important toy models for entanglement spreading in quantum circuits. We then analyze the non-local nature of these invariant operators, both in discrete (operator Schmidt rank) and continuous (entangling power) settings. Our scrutiny reveals that these operators can be used to simulate any bipartite unitary gate via stochastic local operations and classical communication. Furthermore, we establish a one-to-one connection between the set of local diagonal unitary invariant dual unitary operators with maximum entangling power and the set of complex Hadamard matrices. Finally, we discuss distinguishability of unitary operators in the setting of the stated diagonal symmetry. 
\end{abstract}

\maketitle

\tableofcontents

\section{Introduction}

In quantum theory, the time evolution of the state of a closed quantum system is encoded in a unitary operator. This operator, governing the Schr\"odinger equation, inherits the symmetries of the underlying Hamiltonian describing the energy of the system. In an open quantum system setting \cite{breuer2002theory}, when the particular system of interest is coupled to an environment, one needs to consider interaction terms in the Hamiltonian, which lead to a global system-environment unitary evolution operator of a non-product form. Understanding the structure of such bipartite unitary evolution operators is a major theme of foundational quantum theory. 

In this work, we introduce and thoroughly study a new class of bipartite $d\otimes d$ unitary evolution operators $X$ which have the following \emph{invariance property}:
$$\forall U \in \mathcal G,\qquad X = (U \otimes U) X (U \otimes U)^* \quad \text{ or } \quad X = (U \otimes \bar U) X (U \otimes \bar U)^*,$$
where the invariance group $\mathcal G$ is either the \emph{diagonal unitary group} $\mathcal{DU}(d)$ (consisting of diagonal matrices with complex phases on the diagonal) or the \emph{diagonal orthogonal group} $\mathcal{DO}(d)$ (consisting of diagonal sign matrices). In the unitary case, we call the invariant evolution operators \emph{local diagonal unitary invariant} (LDUI) and, respectively, \emph{conjugate local diagonal unitary invariant} (CLDUI). In the orthogonal case, the two conditions above are identical, and we call the operators \emph{local diagonal orthogonal invariant} (LDOI). The invariance property encodes physical symmetries which might be present in the system + environment couple, substantially simplifying the possible evolutions of the total system. Our inspiration comes from the authors' past work \cite{singh2021diagonal}, where the exact same invariance property was required of \emph{quantum states}. The analysis in that work showed that the corresponding class of invariant bipartite quantum states vastly generalized known important examples in quantum information theory, and allowed us to obtain very strong analytical results characterizing the entanglement properties of the states and of the corresponding quantum channels. 

In the current paper, we report similar findings on local diagonal unitary/orthogonal invariant bipartite unitary matrices. This large class of quantum evolution operators enjoys many interesting properties, making it suitable for analytical study and a valuable source for examples. 

We start by providing a detailed analysis of invariant unitary operators, showing how to construct them from elementary building blocks. The presence of this particular structure renders them suitable for concrete analytical study. We then move on to study the effect of fundamental linear operations (such as the partial transposition and the realignment) on our invariant family of operators. This allows us to investigate the intersection of the family of (C)LDUI/LDOI unitary operators with a special class of bipartite unitary gates -- the \emph{dual unitary gates} \cite{bertini2019exact} -- which have been shown to enjoy many interesting properties with regard to entanglement spreading and quantum chaos \cite{bertini2019entanglement}. We are able to construct large, first of their kind families of dual unitary operators in all dimensions, arguing that the class of (C)LDUI/LDOI operators is a rich source of examples in quantum information theory. We then proceed to analyze the non-local properties of (C)LDUI/LDOI unitary operators. In the discrete setting, non locality is characterized by the operator Schmidt rank \cite{Nielson2003opschmidt}. For diagonal invariant unitary matrices, this quantity admits a simple expression, allowing us to not only reprove and strengthen the main results from \cite{Muller2018schmidt}, but to also settle an open question from the same paper. In a nutshell, we prove that an arbitrary bipartite unitary gate can be probabilistically simulated by a real orthogonal LDOI operator via stochastic local operations and classical communication. In the continuous setting, there are several measures of non-locality, such as \emph{operator entanglement}, \emph{entangling power} \cite{Zanardi2000ep}, and \emph{gate typicality} \cite{Jonnadul2017gt}. We derive explicit expressions for all these quantities for (C)LDUI/LDOI unitary matrices. We then establish an intriguing connection between LDUI unitary matrices achieving the maximum entangling power and complex Hadamard matrices, showing that they are in one-to-one correspondence with each other. We finally discuss the problem of distinguishability of (C)LDUI/LDOI unitary matrices. 

Our paper is organized as follows. Section \ref{sec:LDUI-CLDUI-LDOI} contains a review of known results about local diagonal unitary/orthogonal invariant matrices. In Section \ref{sec:LDOI-unitary} we introduce the main classes of bipartite unitary operators that we study, and discuss their general structural properties. Section \ref{sec:dual-unitary} is devoted to the study of invariant dual unitary matrices. Non locality properties are discussed in Sections \ref{sec:non-locality} and \ref{sec:non-locality2}, while distinguishability of the invariant operators is discussed in Section \ref{sec:distinguishability}. We close the paper with a conclusion section, where further directions of study are presented. 

\section{Local diagonal unitary and orthogonal invariant matrices} \label{sec:LDUI-CLDUI-LDOI}

Recently, the authors developed the theory of \emph{local diagonal unitary/orthogonal invariant} bipartite matrices \cite{nechita2021graphical,singh2021diagonal}; this theoretical work already found several applications in quantum information theory \cite{Singh2020entanglement,singh2020ppt2}. Historically, these models of symmetry have been introduced in \cite{chruscinski2006class}, and later studied in \cite{johnston2019pairwise}. Previously, the focus was on quantum states, different notions of positivity, and entanglement properties of bipartite density matrices. In the present work, the focus is shifted on the dynamical aspects: we consider \emph{local diagonal unitary/orthogonal invariant bipartite unitary transformations}. Before studying the specifics of the unitary model, we need to recall the general setting, and the main characterization of local diagonal invariant matrices; we refer to \cite[Section 2]{singh2021diagonal} for the full details. 

We will exclusively be working with the space of $d\times d$ complex matrices $\M{d}$ and the tensor product space $\M{d}\otimes \M{d}$. For a matrix $A\in\M{d}$, $A^\top, \bar{A}$, and $A^\dagger$ denote its transpose, entrywise complex conjugate, and adjoint (conjugate transpose), respectively. The \emph{unitary} groups in $\M{d}$ and $\M{d}\otimes \M{d}$ are denoted by $\U{d}$ and $\UU{d}$, respectively, while the (real) \emph{orthogonal} groups are similarly denoted by $\O{d}$ and $\OO{d}$, respectively.

\begin{definition} \label{def:LDUI-CLDUI-LDOI}
We consider the following local symmetries of matrices $X\in \M{d}\otimes \M{d}$: 
\begin{center}
\centering
{\renewcommand{\arraystretch}{1.2}
\begin{tabular}{|r|l|c|} 
\hline
Acronym    & Symmetry                                   & Condition                                        \\ 
\hline\hline
\emph{LDUI} & local diagonal unitary invariant           & $(U \otimes U) X (U^\dagger \otimes U^\dagger) = X$          \\ 
\hline
\emph{CLDUI} & conjugate local diagonal unitary invariant & $(U \otimes \bar U) X (U^\dagger \otimes U^\top) = X$  \\ 
\hline
\emph{LDOI}  & local diagonal orthogonal invariant        & $(O \otimes O) X (O^\top \otimes O^\top) = X$    \\
\hline
\end{tabular}
}
\end{center}
The conditions above hold for all diagonal unitary matrices $U \in \mathcal{DU}_d$ and diagonal orthogonal matrices $O \in \mathcal{DO}_d$. The subspaces of \emph{LDUI}, \emph{CLDUI} and \emph{LDOI} matrices in $\mathcal{M}_d(\mathbb{C}) \otimes \mathcal{M}_d(\mathbb{C})$ are denoted by $\LDUI_d$, $\CLDUI_d$, and $\LDOI_d$ respectively.
\end{definition}

From the above definition, it is clear that both $\LDUI_d$ and $\CLDUI_d$ are vector subspaces of $\LDOI_d$. Furthermore, it was shown in \cite{singh2021diagonal}
that any $X\in\LDOI_d$ is of the form
 \begin{align}
     X = X^{(3)}_{(A,B,C)} &\coloneqq \includegraphics[align=c,scale=1.2]{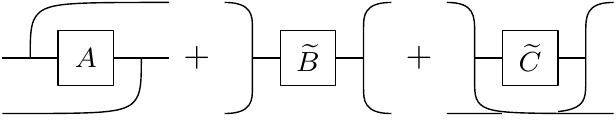} \nonumber \\ 
    &= \sum_{i,j=1}^d A_{ij} \ketbra{ij}{ij} + \sum_{1\leq i\neq j\leq d} B_{ij} \ketbra{ii}{jj} + \sum_{1\leq i\neq j\leq d} C_{ij} \ketbra{ij}{ji}, \label{eq:ABC_LDOI}
 \end{align}
where $A,B,C\in\M{d}$, $\widetilde{B}:=B-\operatorname{diag}B$, $\widetilde{C}:=C-\operatorname{diag}C$, and $\operatorname{diag}A=\operatorname{diag}B=\operatorname{diag}C$. A simple counting of free parameters above implies that $\dim_{\mathbb{C}}\LDOI_d = 3d^2-2d$. 

\begin{remark}
If $B$ (resp. $C$) in Eq.~\eqref{eq:ABC_LDOI} is diagonal, then the resulting family of bipartite matrices form the $\LDUI_d$ (resp. $\CLDUI_d$) subspace. We denote these matrices by
\begin{align} 
    X^{(1)}_{(A,C)} &\coloneqq \includegraphics[scale=1.2,align=c]{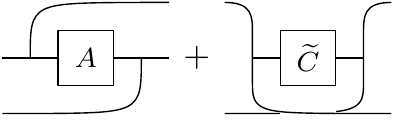} = \sum_{i,j=1}^d A_{ij} \ketbra{ij}{ij} + \sum_{1\leq i\neq j\leq d} C_{ij} \ketbra{ij}{ji},  \\[0.1cm]
    \text{resp. } X^{(2)}_{(A,B)} &\coloneqq \includegraphics[scale=1.2,align=c]{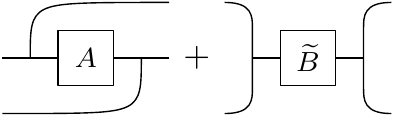} = \sum_{i,j=1}^d A_{ij} \ketbra{ij}{ij} + \sum_{1\leq i\neq j\leq d} B_{ij} \ketbra{ii}{jj} .
\end{align}
\vspace{0.1cm}
Again, a simple counting of free parameters yields $\dim_{\mathbb C}  \CLDUI_d = \dim_{\mathbb C}  \LDUI_d = 2d^2-d$.
\end{remark}

\begin{remark}
In \cite{singh2021diagonal}, the sets of matrix pairs and triples with equal diagonals were denoted as 
\begin{align}
    \MLDUI{d} &\coloneqq \{ (A,B) \in \M{d}\times \M{d} \, \big| \, \operatorname{diag}(A)=\operatorname{diag}(B) \}, \label{eq:MLDUI}\\
    \MLDOI{d} &\coloneqq \{ (A,B,C) \in \M{d}\times \M{d}\times \M{d} \, \big| \, \operatorname{diag}(A)=\operatorname{diag}(B)=\operatorname{diag}(C)\}. \label{eq:MLDOI}
\end{align}
The above sets become vector spaces when equipped with the usual component-wise addition and scalar multiplication. Clearly, the following vector space isomorphisms hold:
\begin{equation*}
    \LDUI_d \cong \CLDUI_d \cong \MLDUI{d} \quad\text{and} \quad \LDOI_d \cong \MLDOI{d}.
\end{equation*}
\end{remark}

\begin{remark}
In this paper, whenever we consider $X^{(3)}_{(A,B,C)}\in \LDOI_d$, the defining matrix triple $(A,B,C)$ is implicitly assumed to lie in $\MLDOI{d}$.
\end{remark}

The $\LDOI_d$ subspace enjoys invariance properties under several operations \cite[Proposition 4.3]{singh2021diagonal}. In addition to transposition and adjoint, we will consider two other operations: realignment and partial transposition, which are defined in Figure~\ref{fig:PT+realign} for $X\in \M{d}\otimes \M{d}$.
\begin{figure}[H]
    \centering
    \includegraphics[scale=1.2]{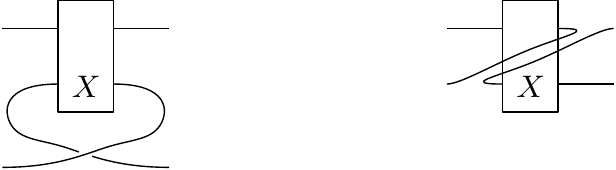}
    \caption{From left to right: the partially transposed $X^{\Gamma} := [\operatorname{id}\otimes \operatorname{transp}](X)$ and the realigned $X^R:=\operatorname{realign}(X)$.}
    \label{fig:PT+realign}
\end{figure}
By defining coordinates $X_{ij,kl}:= \langle ij|X|kl\rangle$, the above definitions are  equivalent to saying
\begin{align*}
        X^\Gamma_{ij, kl} = X_{il, kj} \qquad\text{and}\qquad X^R_{ij, kl} = X_{ik, jl}.
\end{align*}

\begin{proposition}\label{prop:LDOI_symm}
The $\LDOI_d$ subspace is invariant under the operations of transposition, conjugate transposition, realignment, and partial transposition. More precisely, for $X^{(3)}_{(A,B,C)}\in \LDOI_d$, we have 
\begin{alignat*}{2}
    \left(X^{(3)}_{(A,B,C)}\right)^\top &= X^{(3)}_{(A,B^\top,C^\top)} \quad \quad \left(X^{(3)}_{(A,B,C)}\right)^\dagger &&= X^{(3)}_{(\overbar{A},B^\dagger,C^\dagger)} \\
    \left(X^{(3)}_{(A,B,C)}\right)^R &= X^{(3)}_{(B,A,C)} \quad \qquad \left(X^{(3)}_{(A,B,C)}\right)^\Gamma &&= X^{(3)}_{(A,C,B)}.
\end{alignat*}
\end{proposition}

The next result \cite[Proposition 4.1]{singh2021diagonal} shows that matrices in $\LDOI_d$ exhibit an interesting block diagonal structure, which will play a crucial role in later sections when we study the intersection of $\LDOI_d$ with the unitary group $\UU{d}$.

\begin{proposition}\label{prop:LDOI_block}
For every $X^{(3)}_{(A,B,C)}\in\LDOI_d$, the following block decomposition holds:
\begin{align*}
    X^{(3)}_{(A,B,C)} &= B \oplus \left( \bigoplus_{i < j} \begin{bmatrix} A_{ij} & C_{ij} \\ C_{ji} & A_{ji} \end{bmatrix} \right).
\end{align*}
\end{proposition}

In \cite[Lemma 9.3]{singh2021diagonal}, a bilinear operation $\circ$ was defined on the set $\MLDOI{d}$ of matrix triples with equal diagonals: $(A_1,B_1,C_1) \circ (A_2,B_2,C_2) =(\mathfrak{A},\mathfrak{B},\mathfrak{C})$, where 
\begin{align*}
    \mathfrak{A} &= A_1 A_2, \\
    \mathfrak{B} &= B_1 \odot B_2 + C_1\odot C_2^\top + \operatorname{diag}(A_1 A_2 - 2A_1\odot A_2), \\
    \mathfrak{C} &= B_1\odot C_2 + C_1\odot B_2^\top + \operatorname{diag}(A_1 A_2 - 2A_1\odot A_2).
\end{align*}
This operation corresponds to the composition of \emph{diagonal orthogonal covariant} (DOC) linear maps between matrix algebras \cite[Section 6]{singh2021diagonal}. In the $\LDOI_d$ subspace, this corresponds to the tensor contraction depicted in Figure \ref{fig:composition-X-LDOI}, top panel. We now introduce a new bilinear operation on $\MLDOI{d}$ triples, corresponding to the (usual) matrix multiplication of the corresponding matrices in $\LDOI_d$.

\begin{definition}\label{def:product-MLDOI}
    For two triples $(A_1,B_1,C_1), (A_2, B_2, C_2) \in \MLDOI{d}$, define their \emph{product} 
    $$\MLDOI{d} \ni (A,B,C):= (A_1,B_1,C_1) \, \cdot \, (A_2, B_2, C_2)$$ 
    by
    \begin{align*}
    A &= A_1 \odot A_2 + C_1\odot C_2^\top + \operatorname{diag}(B_1 B_2 - A_1\odot A_2 - C_1 \odot C_2^\top), \quad B = B_1 B_2, \\
    \text{and}\quad C &= A_1\odot C_2 + C_1\odot A_2^\top + \operatorname{diag}(B_1 B_2 - A_1\odot C_2 - C_1\odot A_2^\top).
\end{align*}
\end{definition}

\begin{remark}
The diagonal matrices in the above formulas are there to ensure that $A,B$, and $C$ have equal diagonals. These terms can be replaced by the simpler expression $\operatorname{diag}(B_1 B_2 - 2 B_1\odot B_2)$.
\end{remark}

The product ``$\cdot$'' from Definition \ref{def:product-MLDOI} corresponds to the usual matrix multiplication of LDOI matrices (see Figure \ref{fig:composition-X-LDOI}, bottom panel). 

\begin{lemma}\label{lemma:LDOIprod}
Let $X^{(3)}_{(A_1,B_1,C_1)}, X^{(3)}_{(A_2,B_2,C_2)}\in \LDOI_d$. Then, their matrix product
$$X^{(3)}_{(A_1,B_1,C_1)} X^{(3)}_{(A_2,B_2,C_2)}=X^{(3)}_{(A,B,C)},$$ 
where $(A,B,C)= (A_1,B_1,C_1)  \cdot (A_2, B_2, C_2)$.
\end{lemma}
\begin{proof}
The result can be checked either graphically or by using the formula in Eq.~\eqref{eq:ABC_LDOI}. 
\end{proof}

\begin{figure}[H]
    \centering
    \includegraphics[scale=1.1]{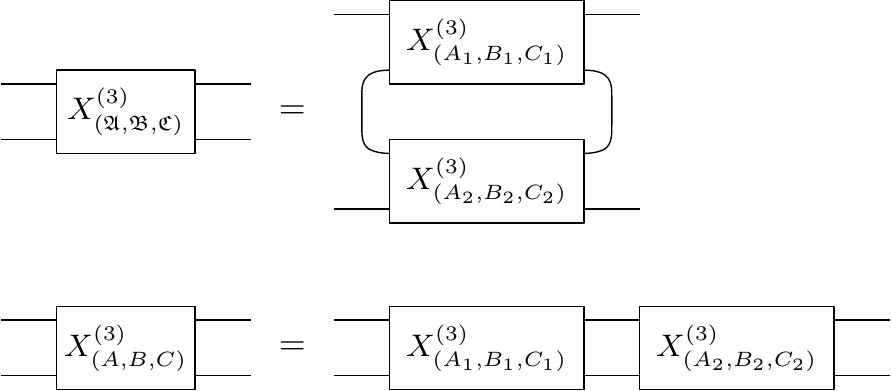}
    \caption{Two binary operations on $\MLDOI{d}$. Top: composition of DOC maps. Bottom: matrix multiplication.}
    \label{fig:composition-X-LDOI}
\end{figure}

\section{Unitary and orthogonal LDOI matrices}\label{sec:LDOI-unitary}

Having introduced the subspaces of local diagonal unitary and orthogonal invariant matrices in the previous section, we shall focus now on matrices which have additional structure: they should be unitary (or orthogonal). Studying unitary operators is motivated by quantum theory, where they model the time evolution of quantum states. The bipartite case is of great importance, since unitary operators acting on a tensor product space model the interaction between two quantum systems. The other extra structure relevant for quantum theory, that of positivity (used to model open quantum systems and density matrices), has been thoroughly investigated in \cite{singh2021diagonal}. 

We start with a simple result, namely the characterization of unitary (resp.~orthogonal) matrices within the $\LDOI_d$ subspace. We will denote the unit circle in the complex plane by $$\mathbb{T}:=\{z\in\mathbb{C} \,\,\vert\,\, |z|=1 \}.$$

\begin{proposition}\label{prop:unitary-condition}
An \emph{LDOI} matrix $X^{(3)}_{(A,B,C)}$ is unitary if and only if
\begin{itemize}
    \item $B$ is unitary,
    \item $\forall i < j$, there exists a phase $\omega_{ij} \in \mathbb T$ such that $A_{ji} = \omega_{ij} \overline{A_{ij}} \text{ and } C_{ji} = - \omega_{ij} \overline{C_{ij}},$
    \item $\forall i < j, \,\, |A_{ij}|^2 + |C_{ij}|^2 = 1$.
\end{itemize}
Orthogonal \emph{LDOI} matrices admit precisely the same characterization, with the field of complex numbers replaced by that of real numbers (and phases $\omega \in \mathbb T$ replaced by signs $\theta = \pm 1$).
\end{proposition}

\begin{proof}
We first note that the identity matrix $\iden_d \otimes \iden_d = X^{(3)}_{(\mathbb J_d, \iden_d, \iden_d)}$, where $\mathbb{J}_d\in \M{d}$ is the all ones matrix. Moreover, since $\left(X^{(3)}_{(A,B,C)}\right)^\dagger = X^{(3)}_{(\bar A, B^\dagger, C^\dagger)}$ (see Proposition~\ref{prop:LDOI_symm}), we can write
\begin{align}
    X^{(3)}_{(A,B,C)}\in \UU{d} &\iff X^{(3)}_{(A,B,C)} \left(X^{(3)}_{(A,B,C)}\right)^\dagger = \iden_d \otimes \iden_d \nonumber \\
    &\iff (A,B,C) \cdot (\bar A, B^\dagger, C^\dagger) = (\mathbb J_d, \mathbb 1_d, \mathbb 1_d).
\end{align}
By Definition~\ref{def:product-MLDOI}, this amounts to saying that $B\in \U{d}$ and 
\begin{align}\label{eq:prop3.1}
\forall i \neq j, \qquad |A_{ij}|^2 + |C_{ij}|^2 = 1 \quad \text{and} \quad A_{ij} \overbar{C_{ji}}  = - C_{ij} \overbar{A_{ji}}.  
\end{align}
Let us now fix some $i \neq j$. If both $A_{ij}$ and $C_{ij}$ are non-zero, we have 
\begin{equation}
\frac{\overbar{A_{ji}}}{A_{ij}} = - \frac{\overbar{C_{ji}}}{C_{ij}} =:z, \text{ for some } z \in \mathbb C.    
\end{equation}
When combined with the other condition in Eq.~\eqref{eq:prop3.1}, this yields $|z|=1$, and the claim follows. If, say $A_{ij} = 0$, then $|C_{ij}| = 1$. Hence, $A_{ji} = 0$; again, the claim follows. 
\end{proof}

\begin{remark}\label{remark:unitary_LDUI}
From Proposition~\ref{prop:unitary-condition}, it is clear that $X^{(1)}_{(A,C)}\in \LDUI_d$ is unitary if and only if 
\begin{itemize}
    \item $\forall i\in [d], \,\, A_{ii}=C_{ii}\in \mathbb{T}$,
    \item $\forall i < j$, there exists a phase $\omega_{ij} \in \mathbb T$ such that $A_{ji} = \omega_{ij} \overline{A_{ij}} \text{ and } C_{ji} = - \omega_{ij} \overline{C_{ij}},$
    \item $\forall i<j, \,\, |A_{ij}|^2 + |C_{ij}|^2 = 1$.
\end{itemize}
Similarly, $X^{(2)}_{(A,B)}\in \CLDUI_d$ is unitary if and only if $B$ is unitary and $A_{ij} \in \mathbb T$ for all $i\neq j$.
\end{remark}

Perhaps a more intuitive understanding of the above result can be obtained by analyzing the block structure of matrices in $\LDOI_d$. Indeed, from Proposition~\ref{prop:LDOI_block}, it is obvious that
\begin{align*}
    X^{(3)}_{(A,B,C)} = B \oplus \left( \bigoplus_{i < j} \begin{bmatrix} A_{ij} & C_{ij} \\ C_{ji} & A_{ji} \end{bmatrix} \right) &\in \UU{d} \iff B\in \U{d} \,\text{ and } \begin{bmatrix} A_{ij} & C_{ij} \\ C_{ji} & A_{ji} \end{bmatrix} \in \U{2} \,\,\,\forall i<j.
\end{align*}
Note that the conditions on $A,C$ given in Proposition~\ref{prop:unitary-condition} are precisely those which ensure that the $2\times 2$ matrices in the above block decomposition are unitary. If we consider $\UU{d}\cap \LDOI_d$ as a subgroup of the full unitary group $\UU{d}$, it is clear that the following group isomorphism holds:
\begin{equation}
    \UU{d}\cap \LDOI_d \,\,\cong \,\,\U{d} \times \bigtimes_{1\leq i<j\leq d} \U{2}. 
\end{equation}
Similar group isomorphisms can be easily obtained for the $\LDUI_d$ and $\CLDUI_d$ classes as well.

\begin{align}
    \UU{d}\cap \LDUI_d \,\,&\cong \,\,\bigtimes_{i=1}^d \U{1} \times \bigtimes_{1\leq i<j\leq d} \U{2}, \\
    \UU{d}\cap \CLDUI_d \,\,&\cong \,\,\U{d} \times \bigtimes_{1\leq i<j\leq d} \U{1}. 
\end{align}

\section{Dual unitary LDOI matrices}\label{sec:dual-unitary}

Dual unitary matrices have been considered recently in relation to various problems in many body physics. Basically, if the nearest neighbor interactions in certain 1+1 many body lattice models are facilitated by dual unitary operators, then correlations between initially localized observables can be explicitly and analytically computed \cite{bertini2019exact,piroli2020exact}. There are several related sets of bipartite unitary operators, relevant to other fields, such as unitary matrices which are still unitary after \emph{partial transposition} \cite{deschamps2016some,benoist2017bipartite}, and \emph{perfect unitaries} \cite{pastawski2015holographic} or \emph{absolutely maximally entangled (AME) states} \cite{helwig2012absolute}. 
 
In this section, we describe the intersection of these special classes of unitary matrices with the $\LDUI_d, \CLDUI_d,$ and $\LDOI_d$ subspaces. We shall see that these properties are translated into certain non-trivial constraints that the $A,B,C$ matrices have to satisfy. We are able to explicitly construct families of dual unitary operators having extra local symmetries, providing several examples of such evolution operators. Let us start by recalling the definitions of the particular classes of bipartite unitary operators that are of interest to us. 

\begin{definition}
    A bipartite unitary operator $U \in \UU{d}$ is called:
    \begin{itemize}
        \item \emph{dual} if its realignment is also a unitary operator: $U^R \in \UU{d}$.
        \item \emph{PT} if its partial transpose is also a unitary operator: $U^\Gamma \in \UU{d}$.
        \item \emph{perfect} if it is both dual and PT: $U^R, U^\Gamma \in \UU{d}$.
    \end{itemize}
    We recall that the realignment and the partial transposition are defined in Figure~\ref{fig:PT+realign}.
    
\end{definition}

Recall from Proposition~\ref{prop:LDOI_symm} that the LDOI subspace is invariant under the operations of realignment and partial transposition:
$$\left(X^{(3)}_{(A,B,C)}\right)^R = X^{(3)}_{(B,A,C)} \quad \text{and} \quad \left(X^{(3)}_{(A,B,C)}\right)^\Gamma = X^{(3)}_{(A,C,B)}.$$
In what follows, we shall focus on the case of dual unitary operators; the case of PT unitaries is very similar, and can be easily deduced from the results in the dual case, see Remark~\ref{remark:dual-PT}.

\begin{proposition}\label{prop:dual}
An \emph{LDOI} matrix $X^{(3)}_{(A,B,C)}$ is dual unitary if and only if
\begin{itemize}
    \item $A$ and $B$ are unitary,
    \item $\forall i<j$, $|A_{ij}|^2 = |B_{ij}|^2 = 1-|C_{ij}|^2$,
    \item $\forall i<j$, there exist complex phases $\omega_{ij}\in \mathbb T$ such that
    $$A_{ji} = \omega_{ij} \overline{A_{ij}}, \qquad B_{ji} = \omega_{ij} \overline{B_{ij}}, \qquad C_{ji} = -\omega_{ij} \overline{C_{ij}}.$$
\end{itemize}
In particular, an \emph{LDUI} matrix $X^{(1)}_{(A,C)}$ is dual unitary if and only if
\begin{equation*}
    \forall i,j\in [d], \,\,\, C_{ij}\in\mathbb{T} \quad\text{and}\quad A = \operatorname{diag}(C).
\end{equation*}
\emph{CLDUI} dual unitary matrices $X^{(2)}_{(A,B)}$ do not exist for $d \geq 3$, while for $d=2$, they are given by $A,B$ with $\operatorname{diag}A=\operatorname{diag}B=0$ and $A_{ij},B_{ij}\in\mathbb{T}$ for $i\neq j$.
\end{proposition}
\begin{proof}
The main claim follows from Proposition \ref{prop:unitary-condition} by demanding that both 
$$X^{(3)}_{(A,B,C)} \quad\text{and}\quad \left( X^{(3)}_{(A,B,C)}\right)^R = X^{(3)}_{(B,A,C)}$$ are unitary. In the LDUI case, the unitary matrix $B$ must be diagonal, hence $B_{ii} \in \mathbb T$ for all $i \in [d]$. Since $|A_{ij}|^2=|B_{ij}|^2$ for all $i\neq j$, the same must hold for the unitary matrix $A$. The only restriction on the off-diagonal elements of $C$ is that $|C_{ij}| = 1$, which proves the claim. 

In the CLDUI case, when $d \geq 3$, a similar line of reasoning shows that the off-diagonal entries of $A$ and $B$ must have modulus one; this is impossible, since each column of $A$ (or $B$) has norm one and at least two off-diagonal elements. We leave the proof of the $2 \times 2$ case to the reader. 
\end{proof}

The canonical example of a dual unitary operator is the \emph{swap} gate $S\in \UU{d}$. It can be easily verified that $S=X^{(1)}_{(A,C)}\in \LDUI_d$, where $C=\mathbb{J}_d$ is the all ones matrix and $A=\operatorname{diag}C$.
\begin{equation}\label{eq:swap}
   \forall \ket{\psi},\ket{\phi}\in \C{d}: \qquad S(\ket{\psi}\otimes \ket{\phi}) = \ket{\phi}\otimes \ket{\psi}.
\end{equation}
\begin{figure}[H]
    \centering
    \includegraphics[scale=1.2]{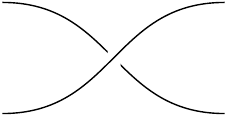}
    \caption{Graphical representation of the swap gate.}
    \label{fig:swap}
\end{figure}

\begin{remark}\label{remark:dual-PT}
By using the block structure of \emph{LDOI} matrices (see Proposition~\ref{prop:LDOI_block}), we can restate the above result as follows. An LDOI matrix $X^{(3)}_{(A,B,C)}$ is dual unitary if and only if
\begin{itemize}
    \item $A$ and $B$ are unitary, 
    \item $\forall i < j$, the $2 \times 2$ matrices
    $\begin{bmatrix}
    A_{ij} & C_{ij} \\
    C_{ji} & A_{ji}
    \end{bmatrix}, \begin{bmatrix}
    B_{ij} & C_{ij} \\
    C_{ji} & B_{ji}
    \end{bmatrix}$ are unitary.
\end{itemize}
Similarly, an \emph{LDOI} matrix $X^{(3)}_{(A,B,C)}$ is \emph{PT} unitary if and only if
\begin{itemize}
    \item $B$ and $C$ are unitary, 
    \item $\forall i < j$, the $2 \times 2$ matrices
    $\begin{bmatrix}
    A_{ij} & C_{ij} \\
    C_{ji} & A_{ji}
    \end{bmatrix}, \begin{bmatrix}
    A_{ij} & B_{ij} \\
    B_{ji} & A_{ji}
    \end{bmatrix}$ are unitary.
\end{itemize}
\end{remark}

While the conditions given in Proposition~\ref{prop:dual} completely characterize the class of dual unitary LDOI matrices, it is difficult to explicitly find matrices $A,B,C\in\M{d}$ that actually fulfill these conditions. This has been a recurring theme in the theory of dual unitary matrices: except for the $d=2$ case (where a complete characterization of dual unitary matrices has been derived \cite{bertini2019exact}), very few explicit examples of dual unitary matrices are known in higher dimensions \cite{arul2020dual, arul2021dual}. In our case, the problem boils down to finding
\begin{enumerate}
    \item two matrices $A,B\in\U{d}$ satisfying $\operatorname{diag}A=\operatorname{diag}B$ and $|A_{ij}|=|B_{ij}|$ for all $i\neq j$,
    \item phases $\omega_{ij}\in \mathbb{T}$ such that $A_{ji} = \omega_{ij}\overbar{A_{ij}}$ and
    $B_{ji} = \omega_{ij}\overbar{B_{ij}}$ for all $i<j$. 
\end{enumerate}
Once we have found $A,B$ satisfying these conditions, it is easy to construct $C$ by simply fixing $C_{ii}=A_{ii}=B_{ii}$ for all $i$ and choosing arbitrary $C_{ij}\in\mathbb{C}$ satsifying $|C_{ij}|^2 = 1 - |A_{ij}|^2 = 1- |B_{ij}|^2$ for all $i<j$. The remaining entries must then be $C_{ji} = -\omega_{ij}\overbar{C_{ij}}$ for all $i<j$. 

We now construct some examples of $A, B\in \M{d}$ that satisfy (1) and (2) above:
\begin{itemize}
    \item For every orthogonal projection $P\in\M{d}$, it is easy to see that $A=B=2P - \iden_d$ satisfy the conditions given in (1) and (2) with $\omega_{ij}=1$ for all $i<j$. Notice that $A=B$ are both Hermitian and unitary in this case. Constructing $C$ as above then gives rise to a family of dual unitary LDOI matrices for each orthogonal projection $P$.
    \item We can generalize the last example as follows. If, for an orthogonal projection $P\in\M{d}$, we choose $A=B=\omega^{1/2} (2P-\iden_d)$ for some $\omega\in \mathbb{T}$, then (1) and (2) will again be satisfied, this time with $\omega_{ij}=\omega$ for all $i<j$. Constructing $C$ in the same way as before will give us a new family of dual unitary LDOI matrices for each orthogonal projection $P$ and phase $\omega$.
    \item It is easy to see that if $A=B=\operatorname{diag}C$ for arbitrary $C$ with $C_{ij}\in\mathbb{T} \,\,\forall i,j$, then we get another family of dual unitary LDOI matrices, parametrized by all $C$ matrices of the stated form. Recall from Proposition~\ref{prop:dual} that this is precisely the class of dual unitary LDUI matrices. These matrices have been constructed in earlier works as well \cite{arul2021dual}, albeit without exploiting their inherent local diagonal unitary invariance property.
\end{itemize}

It would be interesting to see if one can completely characterize pairs of matrices satisfying conditions (1) and (2) above. We leave this as an open problem for the readers.

Our final result in this section illustrates that demanding dual and PT unitarity simultaneously turns out to be too strong a constraint to satisfy in the LDOI subspace. 

\begin{proposition}\label{prop:perfect_LDOI}
No perfect \emph{LDOI} unitary matrices exist.
\end{proposition}
\begin{proof}
Requiring that the three LDOI matrices $$X^{(3)}_{(A,B,C)}, \quad \left(X^{(3)}_{(A,B,C)}\right)^R = X^{(3)}_{(B,A,C)},\quad\text{and}\quad\left( X^{(3)}_{(A,B,C)}\right)^\Gamma = X^{(3)}_{(A,C,B)}$$ be unitary leads us to the following set of conditions:
\begin{itemize}
    \item $A$, $B$, $C$ are unitary,
    \item for all $i<j$, $|A_{ij}|^2 = |B_{ij}|^2 = 1-|C_{ij}|^2$ and $|B_{ij}|^2 = |C_{ij}|^2 = 1-|A_{ij}|^2$,
    \item for all $i<j$, there exists a complex phase $\omega_{ij}\in \mathbb T$ such that
    $$A_{ji} = \omega_{ij} \overline{A_{ij}} \qquad B_{ji} = \omega_{ij} \overline{B_{ij}} \qquad C_{ji} = -\omega_{ij} \overline{C_{ij}},$$ 
    and a complex phase $\lambda_{ij}\in \mathbb T$ such that
    $$B_{ji} = \lambda_{ij} \overline{B_{ij}} \qquad C_{ji} = \lambda_{ij} \overline{C_{ij}} \qquad A_{ji} = -\lambda_{ij} \overline{A_{ij}}.$$ 
\end{itemize}
The second and third conditions above imply that for all $i \neq j$, $|A_{ij}| = |B_{ij}| = |C_{ij}| = 1/\sqrt 2$; in particular, the off-diagonal entries of $A$, $B$, $C$ are non-zero. From the third condition for $B$, we obtain $\omega_{ij} = \lambda_{ij}$, which contradicts the same condition for $A$ and $C$.
\end{proof}

\section{Operator Schmidt rank}\label{sec:non-locality}

In this section, we introduce an important discrete measure of non-locality of bipartite matrices, namely the \emph{operator schmidt rank}, and study it for LDOI unitary matrices. Let us begin with the \emph{operator Schmidt decomposition} of $X\in\M{d}\otimes \M{d}$ \cite{Nielson2003opschmidt}:
\begin{equation}
    X = \sum_{i=1}^\Omega \lambda_i A_i \otimes B_i, \quad\text{where}
\end{equation}
\begin{itemize}
    \item $\Omega(X)=\operatorname{rank}X^R$ is called the \emph{operator Schmidt rank} of $X$.
    \item $\lambda_i>0$ are the (non-zero) singular values of $X^R$ (known as the \emph{Schmidt coefficients} of $X$), 
    \item and $\{A_i\}_{i=1}^\Omega, \{B_i\}_{i=1}^\Omega$ in $\M{d}$ are orthonormal\footnote{With respect to the Hilbert Schmidt inner product defined as $\langle A,B \rangle := \operatorname{Tr}(A^\dagger B)$.}.
\end{itemize}
The number $\Omega(X)$ is the minimum number for which $X$ admits a decomposition of the above form. Clearly, $X$ is of the product form if and only if $\Omega(X)=1$. Higher values of $\Omega(X)$ indicate higher levels of non-locality of $X$, with the maximum being attained for $\Omega(X)=d^2$. Therefore, it is clear that $\Omega(X)$ is a discrete measure of non-locality of $X$. For any non-empty subset $S$ of $\M{d}\otimes \M{d}$, we define the set of all allowed operator Schmidt ranks for matrices in $S$ as
\begin{equation}
    \Omega(S) := \{ \Omega(X)\in \mathbb{N} \, \vert \, X\in S \}\subseteq \{1,2,\ldots ,d^2 \}.
\end{equation}

We have the following result for computing the operator Schmidt ranks of matrices in $\LDOI_d$.

\begin{proposition}\label{prop:LDOI_opschm}
The operator Schmidt rank of $X^{(3)}_{(A,B,C)}\in \LDOI_d$ admits the following expression:
\begin{equation*}
    \Omega\left(X^{(3)}_{(A,B,C)}\right) = \operatorname{rank}A + \sum_{i<j} \operatorname{rank}\left(\begin{array}{c c}
   B_{ij} & C_{ij} \\
   C_{ji} & B_{ji}
\end{array} \right).
\end{equation*}
\end{proposition}
\begin{proof}
The proof follows by noting that (see Propositions~\ref{prop:LDOI_symm} and \ref{prop:LDOI_block})
\begin{equation}
    \Omega\left(X^{(3)}_{(A,B,C)}\right) = \operatorname{rank}\left(X^{(3)}_{(A,B,C)}\right)^R = \operatorname{rank}X^{(3)}_{(B,A,C)} = \operatorname{rank}A + \sum_{i<j} \operatorname{rank}\left(\begin{array}{c c}
   B_{ij} & C_{ij} \\
   C_{ji} & B_{ji}
\end{array} \right).
\end{equation}
\end{proof}

In the bipartite unitary group $\UU{d}$, it is known that different operator Schmidt ranks split the group into a finite number (at most $d^2$) of distinct equivalence classes, where two unitaries $X,X'\in \UU{d}$ are defined to be SLOCC-\emph{equivalent} if both of them can be probabilistically simulated from the other via (stochastic) local operations and classical communication (SLOCC). In other words, $X,X'$ are SLOCC-equivalent if and only if $\Omega(X)=\Omega(X')$ \cite{Dur2002schmidt}. It was shown in the same paper that only three such equivalence classes exist when $d=2$, since for $X\in\UU{2}$, the operator Schmidt rank $\Omega(X)=3$ is forbidden. Whether similar restrictions also arise in higher dimensions was an open question until the authors in \cite{Muller2018schmidt} proved that it was not the case.  

\begin{theorem} \cite{Muller2018schmidt} \label{theorem:muller}
For $X\in \UU{d}$, all operator Schmidt ranks $\Omega(X)\in \{1,2,\ldots ,d^2\}$ are allowed whenever $d\geq 3$. In other words,
$$ \forall d\geq 3: \quad \Omega\big(\UU{d}\big) = \{1,2,\ldots ,d^2 \}.$$
\end{theorem}

Within the context of the present paper, it is natural to ask whether the above theorem holds for bipartite unitary matrices in $\LDOI_d$. In what follows, we prove that this is indeed the case. Moreover, we go a step further and show that the above result holds even for (real) \emph{orthogonal} matrices in $\LDOI_d$. The validity of Theorem~\ref{theorem:muller} for real orthogonal matrices was left as an open problem in \cite{Muller2018schmidt}. In summary, the LDOI constraint allows us to 
\begin{itemize}
    \item provide a simple and more uniform proof of Theorem~\ref{theorem:muller},
    \item strengthen Theorem~\ref{theorem:muller} to work for real orthogonal matrices as well.
\end{itemize}
We should point out that any bipartite unitary matrix $X\in\UU{2}$ is \emph{locally equivalent} to a matrix in $\UU{2}\cap \LDOI_2$ \cite{Kraus2001qubitunitary,Zhang2003qubitunitary}, where we say that $X, X'\in \UU{d}$ are locally equivalent if there exist local unitary matrices $U_i,V_i\in \U{d}$ such that 
\begin{equation}
    X' = (U_1\otimes U_2)X (V_1\otimes V_2).
\end{equation}
Since locally equivalent unitary matrices have equal operator Schmidt ranks (see Remark~\ref{remark:localequiv}), it is clear that the set of allowed Schmidt ranks for the full unitary group $\UU{2}$ is the same as that for the subgroup $\UU{2}\cap \LDOI_2$. For all higher dimensions, we have the following theorem.

\begin{theorem}\label{theorem:opschm_LDOI}
For $X\in \OO{d}\cap \LDOI_d$, all operator Schmidt ranks $\Omega(X)\in \{1,2,\ldots ,d^2\}$ are allowed whenever $d\geq 3$. In other words,
$$ \forall d\geq 3: \quad \Omega\Big(\OO{d}\cap \LDOI_d \Big) = \{1,2,\ldots ,d^2 \}.$$
\end{theorem}

\begin{proof}
Let us start the proof by recalling that for $X^{(3)}_{(A,B,C)}\in \LDOI_d$,
\begin{equation}
    X^{(3)}_{(A,B,C)} \in \OO{d} \iff B\in \O{d} \,\text{ and } \begin{bmatrix} A_{ij} & C_{ij} \\ C_{ji} & A_{ji} \end{bmatrix} \in \O{2} \text{ for all } i<j.
\end{equation}
For every LDOI matrix constructed in this proof, the readers should check that the above conditions on $A,B,C$ are satisfied. Moreover, we know from Proposition~\ref{prop:LDOI_opschm} that
\begin{equation}\label{eq:LDOI_opschm}
    \Omega\left(X^{(3)}_{(A,B,C)}\right) =  \operatorname{rank}A + \sum_{i<j} \operatorname{rank}\left(\begin{array}{c c}
   B_{ij} & C_{ij} \\
   C_{ji} & B_{ji}
\end{array} \right).
\end{equation}
The following proof is split into three parts. \\

\noindent \textbf{Part 1.} Ranks $\Omega(X)\in \{1,2,\ldots ,d\}$ are possible. \\
Choose $A\in\Mreal{d}$ to be a sign matrix ($A_{ij}=\pm 1$) and $B=C=\operatorname{diag}A$, so that 
\begin{equation}
\Omega(X^{(3)}_{(A,B,C)})=\operatorname{rank}A.    
\end{equation}
Note that $\operatorname{rank}(\mathbb{J}_d - 2\iden_d)=d$, where $\mathbb{J}_d\in\M{d}$ is the all ones matrix and $\iden_d\in \M{d}$ is the identity matrix. Now, by duplicating as many columns of $\mathbb{J}_d - 2\iden_d$ as needed, $A$ can be constructed to have any given rank in $\{1,2,\ldots ,d\}$.  \\

\noindent \textbf{Part 2.} Ranks $\Omega(X)\in\{d^2-d,d^2-d+1,\ldots , d^2\}$ are possible. \\
Let $B\in \O{d}$, $A=\operatorname{diag}B$, and $C\in\Mreal{d}$ be such that $\operatorname{diag}C=\operatorname{diag}B$ and $C_{ij}=\pm 1 \,\, \forall i\neq j$. Moreover, let $B$ be such that all its off-diagonal entries are non-zero and there are exactly $m(B)$ zero entries present on the diagonal. Then, Eq.~\eqref{eq:LDOI_opschm} reads
\begin{equation}
    \Omega\left(X^{(3)}_{(A,B.C)}\right) = [d-m(B)] + (d^2-d) = d^2-m(B).
\end{equation}
Except for $(d,m)= (3,2)$ and $(d,m)=(3,3)$, a $B\in \O{d}$ with $m(B)=m$ always exists for all $d\geq 3$ and $m\in \{0,1,\ldots d\}$ \cite[Theorem 5.8]{Bailey2019orthozero}. Hence, $B$ can be appropriately chosen so that all Schmidt ranks in $\{d^2-d,d^2-d+1,\ldots , d^2\}$ are attained if $d\geq 4$. When $d=3$, the same argument shows that the ranks in $\{8,9\}$ are possible. For the remaining ranks, we can easily construct the required matrix triples explicitly:
\begin{align*}
   A = \left(\begin{array}{c r r}
   1 & 0 & 0 \\
   0 & 1 & 1 \\
   0 & \phantom{\pm}1 & \pm 1
\end{array} \right), \,\, B = \operatorname{diag}A, \,\, C = \left(\begin{array}{c r r} 
   1 & 1 & 1 \\
   1 & 1 & 0 \\
   1 & \phantom{\pm}0 & \pm 1
\end{array} \right) &\implies \Omega\left(X^{(3)}_{(A,B,C)}\right) = \begin{cases}
6, &\text{if } A_{33} = 1 \\
7, &\text{if } A_{33} = -1
\end{cases}
\end{align*}
\vspace{0.1cm}

\noindent \textbf{Part 3.} Ranks $\Omega(X)\in\{d+1,d+2,\ldots ,d^2-d-1\}$ are possible. \\
We first deal with dimensions $d\geq 5$. Let $B\in \O{d}$, $C=\operatorname{diag}B$, and $A\in\Mreal{d}$ be such that $\operatorname{diag}A=\operatorname{diag}B$ and $A_{ij}=\pm 1 \,\, \forall i\neq j$. Then, Eq.~\eqref{eq:LDOI_opschm} reads
\begin{equation}
    \Omega\left(X^{(3)}_{(A,B.C)}\right) = \operatorname{rank}A + [d^2 - d - n(B)],
\end{equation}
where $n(B)$ denotes the number of zero entries present in the off-diagonal part of $B$. By suitably permuting the rows/columns of $B$, we can arrange for all its zero entries to be present only in the off-diagonal part. When $d\geq 5$, the total number of zero entries in a $d\times d$ orthogonal matrix can be anything within the set (\cite[Theorem 2.5]{Song2020orthozero})
\begin{align*}
\{0,1,2,\ldots ,d^2-d-4,d^2-d-2,d^2-d\}.
\end{align*}
Hence, by choosing $B\in\O{d}$ to have the appropriate number $n(B)$ of off-diagonal zeros, all Schmidt ranks in $\{d+4,d+5,\ldots ,d^2-d-1\}$ can be attained, irrespective of what $\operatorname{rank}A$ is. For the remaining ranks $\{d+1,d+2,d+3 \}$, let us construct the desired matrices explicitly for the $d=5$ case. We will now choose $B=\operatorname{diag}A=\operatorname{diag}C$, so that according to Eq.~\eqref{eq:LDOI_opschm},
\begin{equation}\label{eq:Bdiag}
    \Omega\left(X^{(3)}_{(A,B,C)}\right) = \operatorname{rank}A + n'(C),
\end{equation}
where $n'(C)=d^2-d-n(C)$ is the total number of non-zero off-diagonal entries of $C$. Then, 
\begin{align*}
    A = \left(\begin{array}{r r r r r}
   -1 & 0 & 1 & 1 & 1 \\
   0 & -1 & 1 & 1 & 1 \\
   1 & 1 & -1 & 1 & 1 \\
   1 & 1 & 1 & -1 & -1 \\
   1 & 1 & 1 & 1 & 1
\end{array} \right), \,\,  C = \left(\begin{array}{r r r r r}
   -1 & 1 & 0 & 0 & 0 \\
   1 & -1 & 0 & 0 & 0 \\
   0 & 0 & -1 & 0 & 0 \\
   0 & 0 & 0 & -1 & 0 \\
   0 & 0 & 0 & 0 & \phantom{-}1
\end{array} \right) &\implies \Omega\left(X^{(3)}_{(A,B.C)}\right) = 6. \\
A = \left(\begin{array}{r r r r r}
   -1 & 0 & 1 & 1 & 1 \\
   0 & -1 & 1 & 1 & 1 \\
   1 & 1 & -1 & 1 & 1 \\
   1 & 1 & 1 & -1 & 1 \\
   1 & 1 & 1 & 1 & -1
\end{array} \right), \,\,  C = \left(\begin{array}{r r r r r}
   -1 & 1 & 0 & 0 & 0 \\
   1 & -1 & 0 & 0 & 0 \\
   0 & 0 & -1 & 0 & 0 \\
   0 & 0 & 0 & -1 & 0 \\
   0 & 0 & 0 & 0 & -1
\end{array} \right) &\implies \Omega\left(X^{(3)}_{(A,B.C)}\right) = 7. \\
A = \left(\begin{array}{r r r r r}
   -1 & 0 & 0 & 1 & 1 \\
   0 & -1 & 1 & 1 & 1 \\
   0 & 1 & -1 & 1 & 1 \\
   1 & 1 & 1 & -1 & -1 \\
   1 & 1 & 1 & 1 & 1
\end{array} \right), \,\,  C = \left(\begin{array}{r r r r r}
   -1 & 1 & 1 & 0 & 0 \\
   1 & -1 & 0 & 0 & 0 \\
   1 & 0 & -1 & 0 & 0 \\
   0 & 0 & 0 & -1 & 0 \\
   0 & 0 & 0 & 0 & \phantom{-}1
\end{array} \right) &\implies \Omega\left(X^{(3)}_{(A,B.C)}\right) = 8.
\end{align*}
The idea above is to keep $n'(C)\in\{2,4\}$ and construct $A$ appropriately to either have $\operatorname{rank}A=d-1$ or $\operatorname{rank}A=d$. These constructions can be easily generalized to higher dimensions $d>5$. Moreover, the same idea of keeping $B$ diagonal while changing $n'(C)$ and $\operatorname{rank}A$ appropriately can be applied to obtain all the required Schimdt ranks in dimensions $d=3,4$ as well. The details can be found in the appendix.

\end{proof}

\begin{remark}
The product unitary matrices in $\LDOI_d$ are of the form $Y\otimes Z$, where $Y,Z\in \U{d}$ are diagonal unitary matrices. When $d=2$, $Y,Z$ can also be of the form (for some $\omega_{ij}\in \mathbb{T}$):
\begin{equation*}
    \left(\begin{array}{c c}
   0 & \omega_{12} \\
   \omega_{21} & 0
\end{array} \right).
\end{equation*}
This result can be derived by analyzing the condition for the operator Schmidt rank of an \emph{LDOI} matrix to be one, see Proposition~\ref{prop:LDOI_opschm}. See also \cite[Proposition 3.1 and Remark 3.2]{singh2021diagonal}.
\end{remark}

In a nutshell, Theorem~\ref{theorem:opschm_LDOI} informs us that the intersection of the full bipartite (real) orthogonal group $\OO{d}\subseteq \UU{d}$ with $\LDOI_d$ is still rich enough to accommodate all levels of discrete non-locality as measured by the operator Schmidt rank. In terms of probabilistic interconvertibility of bipartite unitary operators via SLOCC, the following consequence of Theorem~\ref{theorem:opschm_LDOI} is evident.

\begin{corollary}
Any unitary operator in $\UU{d}$ is \emph{SLOCC}-equivalent to a real orthogonal operator in $\LDOI_d$.
\end{corollary}

\section{Continuous measures of non-locality}\label{sec:non-locality2}
In the last section, we focused on the analysis of a \emph{discrete} measure of non-locality of LDOI unitary matrices, namely the operator Schmidt rank. In this section, we develop and study several \emph{continuous} measures of non-locality of LDOI unitary matrices. As before, let us start with the operator Schmidt decomposition of $X\in \UU{d}$:
\begin{equation}\label{eq:Xschmidt}
    X = \sum_{i=1}^\Omega \lambda_i A_i \otimes B_i,
\end{equation}
where, $\Omega(X)=\operatorname{rank}X^R$ is the operator Schmidt rank of $X$, $\lambda_i>0$ are the (non-zero) Schmidt coefficients of $X$, and the matrix sets $\{A_i\}_{i=1}^\Omega, \{B_i\}_{i=1}^\Omega$ in $\M{d}$ are orthonormal. Then,
\begin{equation}
    \operatorname{Tr}(XX^\dagger) = d^2 \implies \frac{1}{d^2}\sum_{i=1}^\Omega \lambda_i^2= 1.
\end{equation}
Thus, we can use different kinds of entropies of the discrete probability vector $\{\lambda_i^2/d^2 \}_{i=1}^\Omega$ to characterize the non-locality of $X$. In particular, we will be interested in the $2$-\emph{Tsallis entropy}:
\begin{equation}\label{eq:opent}
    E(X) = 1-\frac{1}{d^4}\sum_{i=1}^\Omega \lambda_i^4 = 1 - \frac{1}{d^4} \operatorname{Tr}[(X^RX^{R\dagger})^2] ,
\end{equation}
which we simply refer to as the \emph{operator entanglement} of $X$. It can be easily verified that
\begin{itemize}
    \item $E(X)$ takes values within the interval $[0,1-1/d^2]$ for all $X\in\UU{d}$,
    \item $E(X)=0$ if and only if $X$ is a product of local unitary matrices,
    \item $E(X)=1-1/d^2$ if and only if $X$ is a dual unitary operator.
\end{itemize}


\begin{remark}\label{remark:localequiv}
The Schmidt coefficients stay invariant under local unitary operations, i.e. \emph{locally equivalent} unitary operators $X$ and $X'=(U_1\otimes U_2)X(V_1\otimes V_2)$ (for $U_i,V_i\in\U{d}$) have the same Schmidt coefficients, which implies that $E(X)=E(X')$ as well. This is because if $X$ has a Schmidt decomposition as given in Eq.~\eqref{eq:Xschmidt}, then $X'$ admits the following Schmidt decomposition:
\begin{equation}
    X' = \sum_{i=1}^\Omega \lambda_i A'_i \otimes B'_i,
\end{equation}
where the matrix sets $\{A'_i = U_1 A_i U_3\}_{i=1}^\Omega$ and $\{B'_i=U_2 B_i U_4\}_{i=1}^\Omega$ are again orthonormal.
\end{remark} 

It turns out that the average entanglement generated by a unitary operator $X\in\UU{d}$ when acting upon local product vectors admits a neat expression in terms of the operator entanglement of $X$ and $XS$, where $S\in \UU{d}$ is the swap gate. To elaborate on this further, let us look at the quantity $E(XS)$ in some depth. We have 
\begin{equation}
    E(XS) = 1 - \frac{1}{d^4}\sum_{i=1}^l \mu_i^4 = 1 - \frac{1}{d^4}\operatorname{Tr}[(X^{\Gamma}X^{\Gamma \dagger})^2],
\end{equation}
where $S$ is the swap gate, $l=\Omega(XS)$ is the operator Schmidt rank of $XS$, and $\{\mu_i\}_{i=1}^l$ are the Schmidt coefficients of $XS$. $E(XS)$ stays invariant under local unitary operations, since $XS$ and $X'S = (U_1\otimes U_2)X(U_3\otimes U_4)S = (U_1\otimes U_2)XS(V_1\otimes V_2)$ are locally equivalent (for $U_i,V_i\in\U{d}$), see Remark~\ref{remark:localequiv}. Moreover, this quantity is in some sense complementary to the operator entanglement, since for local unitary operators $U_1,U_2\in \U{d}$, it is evident that 
\begin{equation}
    E(U_1\otimes U_2) = 0 \quad\text{but}\quad E((U_1\otimes U_2) S) = 1 - \frac{1}{d^2},
\end{equation}
and for the swap gate, we have $E(S)=1-1/d^2$ but $E(SS)=E(\iden_d\otimes\iden_d)=0$.

Now, if we choose the linear entropy $\mathcal{E}(\ket{\psi})=1-\operatorname{Tr}[(\operatorname{Tr}_2\rho_\psi)^2]$ as our measure of entanglement for pure bipartite states $\rho_\psi=\ketbra{\psi}$ (here $\ket{\psi}\in\C{d}\otimes\C{d}$ and $\operatorname{Tr}_2$ denotes the partial trace with respect to the second subsystem), we can define the \emph{entangling power} of $X\in\UU{d}$ as \cite{Zanardi2000ep}:
\begin{equation}
    e_p(X) = \left(\frac{d+1}{d-1}\right) \mathbb{E}_{\phi,\psi\sim\operatorname{Haar}}[\mathcal{E}(X\ket{\phi}\ket{\psi})],
\end{equation}
where the unit vectors $\ket{\phi},\ket{\psi}\in\C{d}$ are distributed independently and identically according to the Haar measure on the unit sphere in $\C{d}$. With this definition in hand, it can be shown that
\begin{equation}\label{eq:ep}
    e_p(X) = \frac{1}{E(S)}[E(X)+E(XS)-E(S)],
\end{equation}
see \cite{Zanardi2001ep}. Clearly, $0\leq E(X),E(XS)\leq E(S)\implies 0\leq e_p(X)\leq 1$ and local gates have zero entangling power. It is worthwhile to note that even though the swap gate has maximum operator entanglement, its entangling power is zero since $S\ket{\phi}\ket{\psi} = \ket{\psi} \ket{\phi}$ for all unit vectors $\ket{\phi},\ket{\psi}\in\C{d}$.
Thus, by looking at the entangling power alone, we cannot distinguish between local gates and the swap gate. Hence, it is meaningful to study another linear combination of the quantities $E(X)$ and $E(XS)$ for a given $X\in\UU{d}$, which is known as the \emph{gate typicality} of $X$ \cite{Jonnadul2017gt}:
\begin{equation}\label{eq:gt}
    g_t(X) = \frac{1}{2E(S)}[E(X)-E(XS)+E(S)].
\end{equation}
Simple computations reveal that $0\leq g_t(X)\leq 1$. Moreover, $g_t(X)=0$ if and only if $X$ is a local gate and $g_t(X)=1$ if and only if $X$ is locally equivalent to a swap gate.

Now that we have all the relevant quantities defined, let's get back into $\LDOI_d$ to do some explicit calculations. As the following proposition illustrates, for an LDOI unitary matrix, we can easily compute the expressions for all the above introduced measures of non-locality in terms of the associated matrix triple. 

\begin{proposition}\label{prop:ELDOI}
For an LDOI unitary matrix $X^{(3)}_{(A,B,C)}\in \UU{d}$, the following relations hold:
\begin{align*}
    E\left(X^{(3)}_{(A,B,C)}\right) &= 1 - \frac{1}{d^4}\left[\operatorname{Tr}\big\{(AA^\dagger)^2\big\} + \sum_{i\neq j}\bigg\{\left(|B_{ij}|^2 + |C_{ij}|^2 \right)^2 + |B_{ij}\overbar{C_{ji}} + \overbar{B_{ji}}C_{ij}|^2 \bigg\}\right] \\
    E\left(X^{(3)}_{(A,B,C)}S\right) &= 1 - \frac{1}{d^4}\left[\operatorname{Tr}\big\{(CC^\dagger)^2\big\} + \sum_{i\neq j}\bigg\{\left(|B_{ij}|^2 + |A_{ij}|^2 \right)^2 + |B_{ij}\overbar{A_{ji}} + \overbar{B_{ji}}A_{ij}|^2 \bigg\}\right].
\end{align*}
\end{proposition}
\begin{proof}
Recall from Eq.~\eqref{eq:opent} that
\begin{align*}
    E\left(X^{(3)}_{(A,B,C)}\right) &= 1 - \frac{1}{d^4} \operatorname{Tr}\left[\left(X^{(3)R}_{(A,B,C)}X^{(3)R\dagger}_{(A,B,C)}\right)^2\right]  \\
    &= 1 - \frac{1}{d^4}\operatorname{Tr}\left[\left(X^{(3)}_{(B,A,C)} X^{(3)}_{(\bar{B},A^\dagger,C^\dagger)}\right)^2\right] \\
    &= 1 - \frac{1}{d^4}\operatorname{Tr}\left(X^{(3)}_{(\mathfrak{A},\mathfrak{B},\mathfrak{C})}\right),
\end{align*}
where 
\begin{equation*}
    (\mathfrak{A},\mathfrak{B},\mathfrak{C}) = [(B,A,C)\cdot(\overbar{B},A^\dagger,C^\dagger)]\cdot [(B,A,C)\cdot(\overbar{B},A^\dagger,C^\dagger)].
\end{equation*}
Note that the product ``$\cdot$'' used above was introduced in Definition~\ref{def:product-MLDOI} and we have used Lemma~\ref{lemma:LDOIprod} twice in obtaining the above formula. Now, since $\operatorname{Tr}(X^{(3)}_{(\mathfrak{A},\mathfrak{B},\mathfrak{C})}) = \sum_{i,j}\mathfrak{A}_{ij}$, we can easily obtain the desired result. For the second relation, observe that $X^{(3)}_{(A,B,C)}S=X^{(3)}_{(C,B,A)}$.
\end{proof}

\begin{remark}
Once we have the expressions for $E(X)$ and $E(XS)$ for any $X\in\LDOI_d$, the entangling power and gate typicality can also be easily computed from Eqs.~\eqref{eq:ep} and \eqref{eq:gt}.
\end{remark}

In order to motivate our next result, let us first recall from Proposition~\ref{prop:perfect_LDOI} that perfect unitary matrices do not exist in the LDOI subspace. However, we now prove that in the asymptotic limit of the matrix size $d\rightarrow \infty$, it is possible to get arbitrarily close to the set of perfect unitary matrices within the LDOI subspace.

\begin{theorem}\label{theorem:maxent_dualLDUI}
The maximum entangling power of an LDUI dual unitary matrix $X^{(1)}_{(A,C)}$ is achieved when $C$ is a complex Hadamard matrix, with the maximum value being  
\begin{equation*}
    \max_{\substack{\vspace{0.05cm} \\ X\in \LDUI_d \\ X,X^R\in \UU{d}}} e_p(X) = \frac{d}{d+1} \rightarrow 1 \text{ as } d\rightarrow \infty.
\end{equation*}
\end{theorem}
\begin{proof}
Recall that a dual unitary matrix $X$ has maximal operator entanglement $E(X)=E(S)$, so that we have the following expression for its entangling power:
\begin{equation*}
    e_p(X) = \frac{1}{E(S)}[E(X)+E(XS)-E(S)] = \frac{E(XS)}{E(S)}.
\end{equation*}
Thus, maximization of $e_p(X)$ is the same as maximization of $E(XS)$ for dual unitary matrices. In particular, for $X^{(1)}_{(A,C)}\in\LDUI_d$ with $C_{ij}\in\mathbb{T}$ and $A=\operatorname{diag}C$ (Proposition~\ref{prop:dual}), we have 
\begin{align*}
    E\left(X^{(1)}_{(A,C)}S\right) = 1 - \frac{1}{d^4} \operatorname{Tr}[(CC^\dagger)^2],
\end{align*}
see Proposition~\ref{prop:ELDOI} (recall that $X^{(1)}_{(A,C)}=X^{(3)}_{(A,\operatorname{diag}C,C)}$). Moreover, since 
\begin{equation*}
    \operatorname{Tr}[(CC^\dagger)^2] = \sum_{ij} \left|\sum_k C_{ik}\overbar{C_{jk}} \right|^2 = \sum_{i} \left| \sum_k |C_{ik}|^2 \right|^2 + \sum_{i\neq j} \left|\ldots \right|^2 = d^3 + \sum_{i\neq j} \left|\ldots \right|^2 \geq d^3, 
\end{equation*}
and equality is attained if and only if $CC^\dagger=d\iden_d$ (i.e. $C$ is a complex Hadamard matrix), it is clear that $e_p(X^{(1)}_{(A,C)})$ attains its maximum value (within the LDUI dual unitary class) whenever $C$ is a complex Hadamard matrix. Finally, this maximum value can be easily evaluated:
\begin{equation*}
    e_p\left(X^{(1)}_{(A,C)}\right) = \frac{ E\left(X^{(1)}_{(A,C)}S\right)}{E(S)} = \frac{1 - 1/d}{1 - 1/d^2} = \frac{d}{d+1}.
\end{equation*}
\end{proof}

\begin{remark}
Since complex Hadamard matrices exist for each dimension $d\in\mathbb{N}$, the maximum entangling power stated in Theorem~\ref{theorem:maxent_dualLDUI} is attained for all $d\in\mathbb{N}$.
\end{remark}

Following on the previous remark, in the real case of dual orthogonal LDUI matrices, the question of whether the maximum entangling power is attained or not is more subtle. This is because real Hadamard matrices can only exist when the dimension $d$ is $1, 2$ or a multiple of $4$. Moreover, the existence of real Hadamard matrices of order $d=4n$ (for each $n\in\mathbb{N}$) is a fundamental open problem in combinatorics. The following quantity appears in the proof of Theorem \ref{theorem:maxent_dualLDUI}: for a sign matrix $C \in \mathcal M_d(\{\pm 1\})$, define
\begin{equation}\label{eq:Hada_measure}
\mathfrak h(C) = \operatorname{Tr}[(CC^\dagger)^2] = \sum_{i,j=1}^d |\langle C_{i,\cdot}, C_{j, \cdot} \rangle|^2 = \sum_{i,j=1}^d \left|\sum_{k=1}^d C_{ik}\overbar{C_{jk}} \right|^2,   
\end{equation}
where $C_{i,\cdot}$ denotes the $i$-th row of the matrix $C$. The following lemma was proven alongside Theorem \ref{theorem:maxent_dualLDUI} and motivates calling the quantity $\mathfrak h$ a measure of the ``Hadamardness'' of the sign matrix $C$. 

\begin{lemma}
For a sign matrix $C\in \mathcal M_d(\{\pm 1\})$, we have $\mathfrak h(C)  \geq d^3$, with equality being attained if and only if $C$ is a Hadamard matrix ($CC^\dagger = d\iden_d$). 
\end{lemma}
\begin{proof}
Write the expression in Eq.~\eqref{eq:Hada_measure} as 
$$\mathfrak{h}(C) = \sum_{i=1}^d \left| \sum_{k=1}^d |C_{ik}|^2 \right|^2 + 2\sum_{1 \leq i < j \leq d} \left|\sum_{k=1}^d C_{ik}\overbar{C_{jk}} \right|^2 = d^3 +   2\sum_{1 \leq i < j \leq d} \left|\sum_{k=1}^d C_{ik}\overbar{C_{jk}} \right|^2,$$
with equality if and only if all the terms $\sum_{k=1}^d C_{ik}\overbar{C_{jk}}$ are zero. But this is precisely the condition that the rows of the sign matrix $C$ are orthogonal, i.e.~the matrix $C$ is Hadamard. 
\end{proof}

Clearly, for dimensions for which (real) Hadamard matrices exist, $\min \mathfrak h(C) = d^3$, where the minimum is taken over sign matrices $C$. For dimensions $d\geq 3$ which are not multiples of $4$, we list in Table \ref{tbl:mathfrak-h} the values of the minimum of $\mathfrak h$ obtained by enumerating sign matrices. 

\bgroup
\def\arraystretch{1.5}

\begin{table}[H]
\begin{tabular}{|rccc|}
\hline
\rowcolor[HTML]{9B9B9B} 
\multicolumn{4}{|c|}{\cellcolor[HTML]{9B9B9B}$\min \mathfrak h(C)$ for sign matrices $C$}                                                                                                                                                                  \\ \hline
\rowcolor[HTML]{C0C0C0} 
\multicolumn{1}{|r|}{\cellcolor[HTML]{C0C0C0}Matrix size} & \multicolumn{1}{c|}{\cellcolor[HTML]{C0C0C0}$\min \mathfrak h(C)$} & \multicolumn{1}{c|}{\cellcolor[HTML]{C0C0C0}$\# \mathrm{argmin} \, \mathfrak h$} & $C \in \mathrm{argmin} \, \mathfrak h$ \\ \hline
\multicolumn{1}{|r|}{$d=3$}                               & \multicolumn{1}{c|}{33}                                            & \multicolumn{1}{c|}{6}                                                           &    
\bgroup
\def\arraystretch{1}
$\begin{bmatrix}
1 & \phantom{-}1 & \phantom{-}1 \\
1 & \phantom{-}1 & -1 \\
1 & -1 & \phantom{-}1 
\end{bmatrix}
$
\egroup
\\ \hline
\multicolumn{1}{|r|}{$d=5$}                               & \multicolumn{1}{c|}{145}                                           & \multicolumn{1}{c|}{120}                                                         & 
\bgroup
\def\arraystretch{1}
$\begin{bmatrix}
1 & \phantom{-}1 & \phantom{-}1 & \phantom{-}1 & \phantom{-}1\\
1 & \phantom{-}1 & \phantom{-}1 & -1 & -1\\
1 & \phantom{-}1 & -1 & \phantom{-}1 & -1 \\
1 & -1 & \phantom{-}1 & \phantom{-}1 & -1\\
1 & -1 & -1 & -1 & \phantom{-}1
\end{bmatrix}
$
\egroup
\\ \hline
\multicolumn{1}{|r|}{$d=6$}                               & \multicolumn{1}{c|}{264}                                              & \multicolumn{1}{c|}{28800}                                                            & \bgroup
\def\arraystretch{1}
$\begin{bmatrix}
1 & \phantom{-}1 & \phantom{-}1 & \phantom{-}1 & \phantom{-}1 & \phantom{-}1\\
1 & \phantom{-}1 & \phantom{-}1 & \phantom{-}1 & -1 & -1\\
1 & \phantom{-}1 & \phantom{-}1 & -1 & \phantom{-}1 & -1 \\
1 & \phantom{-}1 & -1 & -1 & -1 & \phantom{-}1\\
1 & -1 & \phantom{-}1 & -1 & -1 & \phantom{-}1\\
1 & -1 & -1 & \phantom{-}1 & \phantom{-}1 & -1
\end{bmatrix}
$
\egroup                                       \\ \hline
\end{tabular}
\caption{Values for $\min \mathfrak h$ over sign matrices for dimensions which are not multiples of $4$. We also print the number of elements achieving the minimal value (in dephased form, with first row and columns being ones), as well as the first such sign matrix in lexicographic order.}
\label{tbl:mathfrak-h}
\end{table}
\egroup

We note that for the odd values of $d$ considered above, $\min \mathfrak h(C) = d^3 + d(d-1)$, obtained for matrices $C$ with the property that all the scalar products between different works have absolute value $1$. We conjecture that this is the case for all odd integers $d$.

\section{Distinguishability of LDOI unitary matrices}\label{sec:distinguishability}

Suppose we know that a quantum operation is equally likely to be either one of the two unitary operators $X_1,X_2\in \U{d}$. Our task is to distinguish between the two equally likely choices. Clearly, $X_1,X_2$ can be \emph{perfectly discriminated} if and only if there exists a probe state $\ket{\psi}\in\C{d}$ ($||\psi ||=1)$ such that $X_1\ket{\psi}$ and $X_2\ket{\psi}$ are orthogonal, i.e. $\langle \psi|X_2^{\dagger}X_1|\psi\rangle=0$, since then (and only then) we can use the Holevo-Helstrom theorem to find an optimal measurement that perfectly discriminates between the states $X_1\ket{\psi}$ and $X_2\ket{\psi}$. Historically, these results were first noted in \cite{Acin2001dist,Ariano2001dist}. We introduce the notion of numerical range to formulate the above discussion.

\begin{proposition}\label{prop:discrim}
For $X\in\M{d}$, we define its \emph{numerical range} as $$N(X):=\left\{\langle\psi |X|\psi\rangle \,\, \vert \,\, \ket{\psi}\in \C{d}, ||\psi ||=1\right\}.$$
Then, $X_1,X_2\in \U{d}$ can be perfectly discriminated if and only if $0\in N(X_2^{\dagger}X_1)$.
\end{proposition}

It is reasonable to think that even if $X_1,X_2\in\U{d}$ cannot be perfectly discriminated in the sense of Proposition~\ref{prop:discrim}, we can still find a suitable $d'$-level ancillary system such that $X_1\otimes \iden$ and $X_2\otimes \iden$ in $\mathcal{U}(d\otimes d')$ can be perfectly discriminated (by using entangled probe states, for instance). Unfortunately, this intuition is flawed, since by using the convexity of the numerical range -- which itself is a consequence of the celebrated Toeplitz-Hausdorff theorem \cite{Toeplitz1918,Hausdorff1919} -- it is easy to show that $N(X\otimes \iden) = N(X)$ for all $X\in \M{d}$.

For normal operators $X\in\M{d}$, one can use the spectral theorem to show that $N(X)$ is simply the convex hull of all the eigenvalues, i.e. $N(X) = \operatorname{conv}\operatorname{spec}X$. In particular, for unitary operators $X\in\U{d}$, if we denote by $\theta(X)$ the length of the smallest arc (in radians) that contains all the eigenvalues of $X$ on the unit circle, then it should be clear that the following equivalences hold:
\begin{equation}
    X_1,X_2\in \U{d} \text{ can be perfectly discriminated } \iff 0\in N(X_2^{\dagger}X_1) \iff \theta(X_2^{\dagger}X_1)\geq \pi.
\end{equation}
Moreover, even if $\theta(X_2^{\dagger}X_1)< \pi$, it is possible to find a finite $k\in\mathbb{N}$ such that $\theta[(X_2^{\dagger}X_1)^{\otimes k}]\geq \pi$. This is because for any $X,Y\in \M{d}$, the spectrum of the tensor product $X\otimes Y$ is precisely the set of all pairwise products of eigenvalues of $X$ and $Y$, which clearly implies that 
\begin{equation}
    \forall k\in\mathbb{N}, \forall X\in\U{d}: \quad \theta(X^{\otimes k})= \min\{k\theta(X), 2\pi \}.
\end{equation}

We can thus formulate a more general version of Proposition~\ref{prop:discrim} as follows.
\begin{proposition}\label{prop:discrim2}
For $X_1,X_2\in\U{d}$, define $$k(X_1,X_2):= \bigg\lceil \frac{\pi}{\theta(X_2^{\dagger}X_1)}  \bigg\rceil,$$ 
where $\lceil \cdot \rceil$ denotes the ceiling function. Then, $k$ is the minimum positive integer such that $X_1^{\otimes k}$ and $X_2^{\otimes k}$ can be perfectly discriminated.
\end{proposition}

For LDOI unitary matrices, we can easily provide a simple upper bound on $k$ as follows.

\begin{proposition}
Let $X^{(3)}_{(A,B,C)},X^{(3)}_{(A',B',C')}$ be \emph{LDOI} unitary matrices and let
\begin{equation*}
    k_{\mathfrak{AC}} := \min_{i<j}\left\{ k\left[ \left(\begin{array}{c c}
   A_{ij} & C_{ij} \\
   C_{ji} & A_{ji}
\end{array} \right),\left(\begin{array}{c c}
   A'_{ij} & C'_{ij} \\
   C'_{ji} & A'_{ji}
\end{array} \right) \right] \right\}.
\end{equation*}
Then, the following bound holds:
\begin{align*}
 k[X^{(3)}_{(A,B,C)},X^{(3)}_{(A',B',C')}] \leq \min\left\{ k(B,B'), k_{\mathfrak{AC}} \right\}.
\end{align*}
\end{proposition}
\begin{proof}
The proof follows simply from Proposition~\ref{prop:discrim2} by noting that the bound 
\begin{equation}
     \theta[X^{(3)}_{A,B,C}] \geq \max\left\{\theta(B),\theta\left[ \left(\begin{array}{c c}
   A_{ij} & C_{ij} \\
   C_{ji} & A_{ji}
\end{array} \right)\right]\right\}
\end{equation}
holds for all LDOI unitary matrices $X^{(3)}_{(A,B,C)}$ and $i<j$. The validity of the stated bound can be checked by computing the eigenvalues of $X^{(3)}_{(A,B,C)}$, which are given by (see Proposition~\ref{prop:LDOI_block})
\begin{equation}
    \operatorname{spec} X^{(3)}_{(A,B,C)} = \operatorname{spec} B \cup \bigcup_{i<j} \operatorname{spec} \begin{bmatrix} A_{ij} & C_{ij} \\ C_{ji} & A_{ji} \end{bmatrix}.
\end{equation}
\end{proof}

We should emphasize that the above proposition deals with the problem of \emph{global} discrimination of bipartite unitary operators. More precisely, we have assumed that both subsystems upon which the bipartite operators $X_1,X_2\in \UU{d}$ act are in the possession of a single party, thus enabling the use of entangled probe states $\ket{\psi}\in \C{d}\otimes \C{d}$ in the discrimination process. It is thus quite natural to ask what happens if the underlying subsystems are in the control of different spatially separated parties, so that they only have local resources at their disposal to conduct the discrimination process. This question was first considered and solved in \cite{Duan2008localNR}.

\begin{proposition}
For $X\in\M{d}\otimes \M{d}$, we define its \emph{local numerical range} as $$N^{\otimes}(X):=\left\{\langle\psi \otimes \phi |X|\psi\otimes \phi\rangle \, : \, \ket{\psi},\ket{\phi}\in \C{d}, ||\psi ||=||\phi ||=1\right\}.$$
Then, $X_1,X_2\in \UU{d}$ can be locally perfectly discriminated if and only if $0\in N^{\otimes}(X_2^{\dagger}X_1)$.
\end{proposition}

For LDOI matrices, we now illustrate how one can obtain an expression for the local numerical range in terms of the defining matrix triple $(A,B,C)$.

\begin{proposition}
For $X^{(3)}_{(A,B,C)}\in \LDOI_d$, the following expression holds:
\begin{equation*}
    N^{\otimes}\left(X^{(3)}_{(A,B,C)}\right) = \left\{ \langle v\odot \overbar v | A | w\odot \overbar w \rangle + \langle v\odot w | \widetilde B | v\odot w \rangle + \langle v\odot \overbar w | \widetilde C | v\odot \overbar w \rangle \, : \, \ket{v},\ket{w}\in \C{d}, ||v||=||w||=1 \right\},
\end{equation*}
where $\widetilde B = B-\operatorname{diag}B$, $\widetilde C=C-\operatorname{diag}C$, and $\odot$ represents entrywise multiplication of vectors.
\end{proposition}

\begin{proof}
A simple calculation reveals that for any $\ket{v},\ket{w}\in \C{d}$ (see Eq.~\eqref{eq:ABC_LDOI}),
\begin{align*}
    \langle v \otimes w | X^{(3)}_{(A,B,C)}|v\otimes w\rangle &= \sum_{i,j} A_{ij} |v_i|^2 |w_j|^2 + \sum_{i\neq j} \left(B_{ij} \overbar{v_i w_i} v_j w_j + C_{ij} \overbar{v_i} w_i v_j \overbar{w_j} \right) \\
    &= \langle v\odot \overbar v | A | w\odot \overbar w \rangle + \langle v\odot w | \widetilde B | v\odot w \rangle + \langle v\odot \overbar w | \widetilde C | v\odot \overbar w \rangle
\end{align*}

\end{proof}

Many of the nice properties of the normal numerical range $N(X)$, such as convexity, no longer hold for the local numerical range $N^{\otimes}(X)$ \cite{Duan2008localNR,Puchaa2011localNR}, thus making the problem of local distinguishability more difficult. We expect that the same difficulties would carry over to the LDOI case as well. While respecting the LDOI unitarity constraints on matrix triples $(A,B,C)$ (see Proposition~\ref{prop:unitary-condition}), we leave the problem of determining if $0\in N^{\otimes}\left(X^{(3)}_{(A,B,C)}\right)$ open for the readers.

\section{Conclusion}
Time evolution of isolated quantum systems is governed by unitary operations on the state space of the system. Hence, studying the structure of such evolutions is crucial to the understanding of the dynamical aspects of quantum theory. In this paper, we have looked at the intersection of the bipartite $d\otimes d$ unitary group with the symmetry class of local diagonal unitary/orthogonal invariant (LDUI/LDOI) matrices. We have completely characterized this intersection as a subgroup of the full bipartite unitary group. We have further studied special classes of unitary matrices -- such as dual and PT unitary matrices -- within the said intersection. Notably, we have constructed explicit parametrizations of several large families of LDOI dual unitary matrices in arbitrary dimensions. This is interesting because at the moment, explicit examples of dual unitary matrices are very scarce in the literature. Furthermore, we have shown that although no perfect unitary matrices exist in the LDOI class, it is still possible to get arbitrarily close to such unitary matrices while staying within the LDOI class in the limit of $d\rightarrow \infty$. Our analysis of the nonlocal properties of LDOI unitary matrices reveal that any arbitrary bipartite unitary matrix can be simulated by a real orthogonal LDOI matrix via (stochastic) local operations and classical communication with non-zero probability. Finally, we have found simple expressions of several natural measures of non-locality -- such as the operator Schmidt rank, entangling power, and gate typicality -- for unitary matrices in the LDOI class. Intriguingly, LDUI unitary matrices with maximum entanglement generation capacity are shown to be in one-to-one correspondence with complex Hadamard matrices.

Several open problems stem from our research, some of which are listed below. \begin{itemize}
    \item It would be interesting to obtain a complete parameterization of matrix triples $(A,B,C)$ that satisfy the LDOI dual unitarity conditions as stated in Theorem~\ref{prop:dual}.
    \item An analogue of Theorem~\ref{theorem:maxent_dualLDUI} for the larger class of LDOI dual unitary matrices would be desirable to obtain. Since the constraints defining the LDOI dual unitary class are not well understood at the moment, we expect that maximizing the entangling power while respecting these constraints would be difficult.
    \item Further analysis of the non-local properties of LDOI unitary operators is required to ascertain whether, for instance, universal entanglers \cite{Chen2008universal} exist within the LDOI class. 
    \item It would be compelling to properly analyze the function $\mathfrak{h}: \mathcal M_d(\{\pm 1\})\to \mathbb{R}$ (see Eq.~\eqref{eq:Hada_measure}) as a measure of the ``Hadamardness'' of sign matrices.
    \item A detailed analysis of the product numerical range of LDOI operators is needed to fully understand the problem of local distinguishability of LDOI unitary operators.
\end{itemize}

\appendix

\section{Proof of Theorem~\ref{theorem:opschm_LDOI}}

\noindent\textit{Proof of Theorem~\ref{theorem:opschm_LDOI} continued}. For dimensions $d=3,4$, we now explicitly construct LDOI matrices with Schmidt ranks in the set $\{d+1,d+2,\ldots ,d^2-d-1\}$. We will always take $B=\operatorname{diag}A=\operatorname{diag}B$ unless stated otherwise. For all the constructed matrices, Eq.~\eqref{eq:LDOI_opschm} can be readily used to compute the Schmidt ranks. Let us first deal with the $d=3$ case.
\begin{align*}
   A = \left(\begin{array}{c r r}
   \frac{1}{\sqrt{2}} & \phantom{-}\frac{1}{\sqrt{2}} & -1 \\
   \frac{1}{\sqrt{2}} & \frac{1}{\sqrt{2}} & 1 \\
   1 & -1 &  1
\end{array} \right), \,\, B = C = \left(\begin{array}{c r r} 
   \frac{1}{\sqrt{2}} & -\frac{1}{\sqrt{2}} & 0 \\
   \frac{1}{\sqrt{2}} & \frac{1}{\sqrt{2}} & 0 \\
   0 & 0 & \phantom{-}1
\end{array} \right)  &\implies \Omega\left(X^{(3)}_{(A,B,C)}\right) = 4, \\
A = \left(\begin{array}{c r r}
   1 & 0 & 1 \\
   0 & 1 & 1 \\
   1 & 1 & 1
\end{array} \right), \,\, C = \left(\begin{array}{c r r}
   1 & 1 & 0 \\
   1 & 0 & 0 \\
   0 & 0 & 1 
\end{array} \right) &\implies \Omega\left(X^{(3)}_{(A,B,C)}\right) = 5.
\end{align*}
The constructions for $d=4$ are going to be slightly more tedious.
\begin{align*}
A = \left(\begin{array}{r r r r}
   -1 & 0 & 1 & 1  \\
   0 & -1 & 1 & 1  \\
   1 & 1 & -1 & -1  \\
   1 & 1 & 1 & 1  
\end{array} \right), \,\,  C = \left(\begin{array}{r r r r}
   -1 & 1 & 0 & 0 \\
   1 & -1 & 0 & 0 \\
   0 & 0 & -1 & 0 \\
   0 & 0 & 0 & \phantom{-}1 
\end{array} \right) &\implies \Omega\left(X^{(3)}_{(A,B.C)}\right) = 5.  \\
A = \left(\begin{array}{r r r r}
   -1 & 0 & 1 & 1  \\
   0 & -1 & 1 & 1  \\
   1 & 1 & -1 & 1  \\
   1 & 1 & 1 & -1  
\end{array} \right), \,\,  C = \left(\begin{array}{r r r r}
   -1 & 1 & 0 & 0 \\
   1 & -1 & 0 & 0 \\
   0 & 0 & -1 & 0 \\
   0 & 0 & 0 & -1 
\end{array} \right) &\implies \Omega\left(X^{(3)}_{(A,B.C)}\right) = 6. \\
A = \left(\begin{array}{r r r r}
   -1 & 0 & 0 & 1  \\
   0 & -1 & -1 & 1  \\
   0 & 1 & 1 & 1  \\
   1 & 1 & 1 & -1  
\end{array} \right), \,\,  C = \left(\begin{array}{r r r r}
   -1 & 1 & 1 & 0 \\
   1 & -1 & 0 & 0 \\
   1 & 0 & \phantom{-}1 & 0 \\
   0 & 0 & 0 & -1 
\end{array} \right) &\implies \Omega\left(X^{(3)}_{(A,B.C)}\right) = 7.  \\
A = \left(\begin{array}{r r r r}
   -1 & 0 & 0 & 1  \\
   0 & -1 & 1 & 1  \\
   0 & 1 & -1 & 1  \\
   1 & 1 & 1 & -1  
\end{array} \right), \,\,  C = \left(\begin{array}{r r r r}
   -1 & 1 & 1 & 0 \\
   1 & -1 & 0 & 0 \\
   1 & 0 & -1 & 0 \\
   0 & 0 & 0 & -1 
\end{array} \right) &\implies \Omega\left(X^{(3)}_{(A,B.C)}\right) = 8. \\
A = \left(\begin{array}{r r r r}
   -1 & 0 & 0 & 0  \\
   0 & -1 & 1 & 1  \\
   0 & 1 & -1 & -1  \\
   0 & 1 & 1 & 1  
\end{array} \right), \,\,  C = \left(\begin{array}{r r r r}
   -1 & 1 & 1 & 1 \\
   1 & -1 & 0 & 0 \\
   1 & 0 & -1 & 0 \\
   1 & 0 & 0 & \phantom{-}1 
\end{array} \right) &\implies \Omega\left(X^{(3)}_{(A,B.C)}\right) = 9. \\
A = \left(\begin{array}{r r r r}
   -1 & 0 & 0 & 0  \\
   0 & -1 & 1 & 1  \\
   0 & 1 & -1 & 1  \\
   0 & 1 & 1 & -1  
\end{array} \right), \,\,  C = \left(\begin{array}{r r r r}
   -1 & 1 & 1 & 1 \\
   1 & -1 & 0 & 0 \\
   1 & 0 & -1 & 0 \\
   1 & 0 & 0 & -1 
\end{array} \right) &\implies \Omega\left(X^{(3)}_{(A,B.C)}\right) = 10.
\end{align*}
To obtain Schmidt rank $11$, we must break the above pattern and make $B$ non-diagonal.

\begin{align*}
    A = \left(\begin{array}{r r r r}
   \frac{1}{\sqrt{2}} & \phantom{-}\frac{1}{\sqrt{2}} & 0 & 0  \\
   \frac{1}{\sqrt{2}} & \frac{1}{\sqrt{2}} & 0 & 1  \\
   0 & 0 & 1 & 1  \\
   0 & 1 & \phantom{-}1 & \phantom{-}1  
\end{array} \right), \,\,  B = \left(\begin{array}{r r r r}
   \frac{1}{\sqrt{2}} & -\frac{1}{\sqrt{2}} & 0 & 0  \\
   \frac{1}{\sqrt{2}} & \frac{1}{\sqrt{2}} & 0 & 0  \\
   0 & 0 & \phantom{-}1 & 0 \\
   0 & 0 & 0 & \phantom{-}1 
\end{array} \right), \,\,
&C = \left(\begin{array}{r r r r}
   \frac{1}{\sqrt{2}} & -\frac{1}{\sqrt{2}} & 1 & 1  \\
   \frac{1}{\sqrt{2}} & \frac{1}{\sqrt{2}} & 1 & 0  \\
   1 & 1 & \phantom{-}1 & 0 \\
   1 & 0 & 0 & \phantom{-}1 
\end{array} \right) \\ 
&\implies \Omega\left(X^{(3)}_{(A,B.C)}\right) = 11. \qed
\end{align*}

\bigskip

\bibliographystyle{alpha}
\bibliography{references}

\vspace{1cm}
\bigskip
\bigskip

\hrule
\bigskip

\end{document}